\documentclass[11pt]{article}

\usepackage[letterpaper,margin=1.00in]{geometry}
\usepackage{amsmath, amssymb, amsthm, thmtools, amsfonts, dsfont}
\usepackage{bbm}

\usepackage{ifthen}

\usepackage{tikz}
\usetikzlibrary{positioning,decorations.pathreplacing}

\usepackage{cite}
\usepackage{appendix}
\usepackage{graphicx}
\usepackage{color}
\usepackage{algorithm}
\usepackage[noend]{algpseudocode}
\usepackage{epstopdf}
\usepackage[textsize=tiny]{todonotes}

\usepackage{framed}
\usepackage[framemethod=tikz]{mdframed}
\usepackage[bottom]{footmisc}
\usepackage[shortlabels]{enumitem}
\setitemize{noitemsep,topsep=3pt,parsep=3pt,partopsep=3pt}
\usepackage[font=small]{caption}
\usepackage{xspace}

\usepackage{mathtools}

\newtheorem{theorem}{Theorem}[section]
\newtheorem{lemma}[theorem]{Lemma}
\newtheorem{meta-theorem}[theorem]{Meta-Theorem}

\newtheorem{definition}[theorem]{Definition}

\newcommand{\FullOrShort}{full}

\ifthenelse{\equal{\FullOrShort}{full}}{
    
  \newcommand{\fullOnly}[1]{#1}
  \newcommand{\shortOnly}[1]{}
  }{
    \newcommand{\fullOnly}[1]{}
    \newcommand{\shortOnly}[1]{#1}
  }

\definecolor{darkgreen}{rgb}{0,0.5,0}
\usepackage{hyperref}
\hypersetup{
    unicode=false,          % non-Latin characters in Acrobat’s bookmarks
    colorlinks=true,        % false: boxed links; true: colored links
    linkcolor=red,          % color of internal links (change box color with linkbordercolor)
    citecolor=darkgreen,        % color of links to bibliography
    filecolor=magenta,      % color of file links
    urlcolor=cyan           % color of external links
}
\usepackage[capitalize, nameinlink]{cleveref}
\crefname{theorem}{Theorem}{Theorems}
\Crefname{lemma}{Lemma}{Lemmas}
\Crefname{observation}{Observation}{Observations}
\Crefname{equation}{}{}

\algnewcommand\algorithmicswitch{\textbf{switch}}
\algnewcommand\algorithmiccase{\textbf{case}}

\algdef{SE}[SWITCH]{Switch}{EndSwitch}[1]{\algorithmicswitch\ #1\ \algorithmicdo}{\algorithmicend\ \algorithmicswitch}%
\algdef{SE}[CASE]{Case}{EndCase}[1]{\algorithmiccase\ #1}{\algorithmicend\ \algorithmiccase}%
\algtext*{EndSwitch}%
\algtext*{EndCase}%
\usepackage{thmtools} 
\usepackage{thm-restate}

\usepackage[textsize=tiny]{todonotes}

\usepackage[capitalize, nameinlink]{cleveref}
\crefname{theorem}{Theorem}{Theorems}
\Crefname{lemma}{Lemma}{Lemmas}
\Crefname{observation}{Observation}{Observations}
\Crefname{equation}{}{}

\newcommand{\eps}{\varepsilon}

\newcommand{\card}[1]{\left| #1 \right|}

\newcommand{\eqdef}{{\stackrel{\rm def}{=}}}
\newcommand{\poly}{{\rm poly}}

\newcommand{\floor}[1]{\left\lfloor #1 \right\rfloor}
\newcommand{\ceil}[1]{\left\lceil #1 \right\rceil}

\providecommand{\E}{{\rm \mathbb{E}}}
\providecommand{\Var}{{\rm \mathbb{V}ar}}
\providecommand{\Cov}{{\rm Cov}}

\newcommand*{\mpc}{\textsf{MPC}}
\newcommand*{\local}{\textsf{LOCAL}}
\newcommand*{\congest}{\textsf{CONGEST}}

\renewcommand{\paragraph}[1]{\vspace{0.15cm}\noindent {\bf #1}:}

\usepackage{setspace}

\usepackage{mathtools}
\DeclarePairedDelimiter{\abs}{\lvert}{\rvert}

\makeatletter
\let\oldabs\abs
\def\abs{\@ifstar{\oldabs}{\oldabs*}}
\makeatother

\date{}

\title{\huge Massively Parallel Ruling Set Made Deterministic}

\author{
    Jeff Giliberti\\
    \small{University of Maryland, USA}\\
    \small{\texttt{jeffgili@umd.edu}}
    \and
    Zahra Parsaeian\\
    \small{University of Freiburg, Germany}\\
    \small{\texttt{zahrap@cs.uni-freiburg.de}}
    \and
}

\date{}

\begin{document}
\maketitle

\setcounter{page}{0}
\thispagestyle{empty}
\begin{abstract}
    We study the \textit{deterministic complexity} of the $2$-Ruling Set problem in the model of Massively Parallel Computation (MPC) with linear and strongly sublinear local memory. 

    \textbf{Linear MPC}: We present a constant-round deterministic algorithm for the $2$-Ruling Set problem that matches the randomized round complexity recently settled by  Cambus, Kuhn, Pai, and Uitto~[DISC'23], and improves upon the deterministic $O(\log \log n)$-round algorithm by Pai and Pemmaraju~[PODC'22]. 
    Our main ingredient is a simpler analysis of CKPU's algorithm based solely on bounded independence, which makes its efficient derandomization possible.

    \textbf{Sublinear MPC}: We present a  deterministic algorithm that computes a $2$-Ruling Set in $\tilde O(\sqrt{\log n})$ rounds deterministically. Notably, this is the first deterministic ruling set algorithm with sublogarithmic round complexity, improving on the $O(\log \Delta + \log \log^* n)$-round complexity that stems from the deterministic MIS algorithm of Czumaj, Davies, and Parter~[TALG'21]. Our result is based on a simple and fast randomness-efficient construction that achieves the same sparsification as that of the randomized $\tilde O(\sqrt{\log n})$-round LOCAL algorithm by Kothapalli and Pemmaraju~[FSTTCS'12]. 
\end{abstract}
\newpage
\section{Introduction}
In this paper, we present faster deterministic parallel algorithms for finding $2$-ruling sets. 
Given an $n$-vertex $m$-edge graph $G = (V, E)$ and an integer $\beta \ge 1$, the more general problem of $\beta$-ruling sets consists of finding a subset $S \subseteq V$ of non-adjacent vertices such that each vertex $v \in V \setminus S$ is at most $\beta$ hops away from some vertex in $S$. Thus, a $\beta$-ruling set is also a $\beta+1$ ruling set. This concept serves as a natural generalization of one of the most central and well-studied problems in distributed graph algorithms, known as  \textit{Maximal Independent Set} (MIS), which corresponds to a $1$-ruling set. Generally, for $\beta \ge 1$, the complexity of a $\beta$-ruling set reduces as the value of $\beta$ increases. 

We design $2$-ruling set algorithms for the model of Massively Parallel Computation (\mpc) in the strongly sublinear and linear memory regimes. 
The study of $2$-ruling sets is motivated by its close relationship with MIS, while still permitting the development of considerably faster algorithms.
Additionally, it is known that for problems utilizing MIS as a subroutine, a $\beta$-ruling set may serve as an alternative for some $\beta > 1$ \cite{BBKO22}.

\subparagraph{\mpc\ Model}
Initially introduced by \cite{KSV10} and later refined in \cite{ANOY13, BKS13, GSZ11}, this model is characterized by a set of $M$ machines each with memory $S$. The input is distributed across machines and the computation proceeds in synchronous rounds. Each round machines perform arbitrary local computation and all-to-all communication, sending and receiving up to $S$ words. The main goal is to minimize the number of communication rounds required by the algorithm. A second goal is to minimize the global space needed to solve the problem, i.e., the number of machines times the local memory per machine, which is $\Omega(n+m)$ for graph problems. In the linear regime of \mpc\ each machine is assigned local memory $S = O(n)$, while in the (strongly) sublinear regime of \mpc\ the local memory is $O(n^\alpha)$, for constant $0 < \alpha < 1$.

\subparagraph{Linear MPC} In the linear model of \mpc, a series of works showed that several fundamental problems such as $(\Delta+1)$-coloring \cite{CFG+19, CDP20} and minimum-spanning tree \cite{Now21} admit constant-round deterministic algorithms. Surprisingly, a recent work of \cite{CKPU23} provides a randomized $2$-ruling set algorithm with constant-round complexity improving on the $O(\log \log \log n)$ time algorithm by \cite{HPS14} and the $O(\log \log \Delta)$ time bound that stems from the MIS algorithm by \cite{GGKMR18}. On the deterministic side, \cite{PP22} gave an algorithm that computes a $2$-ruling set in $O(\log \log n)$ time, which improved on the $O(\log \Delta + \log \log^* n)$ round complexity due to the deterministic MIS algorithm of \cite{CDP21mis, CDP21lb, CDP24}.
Key challenges in this domain lie in determining the existence of deterministic algorithms achieving constant-round complexity for $2$-ruling sets and sublogarithmic-round complexity for MIS.

\subparagraph{Sublinear MPC} In the sublinear model of \mpc, the above $O(\log \Delta + \log \log^* n)$-round algorithm by \cite{CDP21mis, CDP21lb} is the fastest known for both MIS and $2$-ruling set. On the randomized side, \cite{GU19} show that MIS can be solved in $\tilde{O}(\sqrt{\log \Delta} + \log \log n)$ rounds and \cite{PP22} show that $2$-ruling set can be solved in $\tilde O(\log^{1/6} \Delta + \log\log n)$, where the $\tilde O(\cdot)$ notation hides $\poly \log (\cdot)$ factors. It may be worth noting that if we limit the global space to $\tilde O(n+m)$, then the fastest $2$-ruling set algorithm has $\tilde O(\log^{1/4} n + \log \log n)$ randomized complexity \cite{PP22} and $O(\log \Delta \log \log n)$ deterministic complexity \cite{CDP21mis, FGG23}.

\subparagraph{Other Related Work}
There is a large body of work studying ruling sets in the LOCAL model \cite{GV7, BHP12, HPS14, SEW13, BKP14, BEPS16}.
The most relevant to ours is the randomized LOCAL algorithm of \cite{KP12} for computing $2$-ruling sets that combined with \cite{Gha16} yields a LOCAL round complexity of $\tilde O (\sqrt{\log n})$.  
On the hardness side, in the LOCAL model, there is a lower bound for $2$-ruling set of $\Omega(\min\{\sqrt{\Delta}, \log_{\Delta}n\})$ deterministic rounds and of $\Omega(\min\{\sqrt{\Delta}, \log_{\Delta}\log n)$ randomized rounds by \cite{BBO22, BBKO22}, which, in terms of its proportion to $n$, are $\Omega(\frac{\log n}{\log \log n})$, and $\Omega(\frac{\log\log n}{\log\log \log n})$, respectively.
For MIS and maximal matching (MM), the best known deterministic lower bound is $\Omega(\min\{\Delta, \log_{\Delta}n\})$ by \cite{BBH+19}, and the best known randomized lower bounds are $\Omega(\min\{\Delta, \log_{\Delta}\log n\})$ by \cite{BBH+19} and $\Omega(\min\{\frac{\log \Delta}{\log \log \Delta}, \log_{\Delta}n\})$ by \cite{KMW16}, which, in terms of its proportion to $n$, are $\Omega(\frac{\log n}{\log \log n})$, $\Omega(\frac{\log\log n}{\log\log \log n})$, and $\Omega(\sqrt{\frac{\log n}{\log \log n}})$, respectively. Via the MPC conditional lower-bound framework by \cite{GKU19, CDP21lb}, these results give the following component-stable lower bounds for sublinear \mpc\ algorithms:
\begin{itemize}
    \item $\Omega(\log \log n)$ for deterministic $2$-ruling set, deterministic and randomized MIS and MM.
    \item $\Omega(\log \log \log n)$ for randomized $2$-ruling set.
\end{itemize} 

\subsection{Our Contribution}
We design improved deterministic algorithms for the problem of $2$-ruling set in the \mpc\ setting with linear and sublinear local memory. 

\subparagraph{Linear MPC Regime} We develop a deterministic algorithm that matches the constant-round complexity of \cite{CKPU23} and even its optimal global space usage.
\begin{theorem}\label{thm:linear-memory}
There is a $O(1)$-round linear \mpc\ algorithm that computes a $2$-ruling set deterministically using linear global space. 
\end{theorem}
Prior to our work, the best known deterministic complexity was $O(\log \log n)$ by a result of \cite{PP22}.
Our algorithm (\Cref{sec:linear}) is obtained by derandomizing the $O(1)$-round algorithm of \cite{CKPU23}. While the derandomization framework of our algorithm has been applied successfully to numerous \mpc\ graph problems \cite{CPS17, CDP20, CDP21col, CDP21mis, CC22, FGG22, FGG23, PP22}, the main challenge lies in analyzing (a slight variation of) \cite{CKPU23}'s algorithm under limited independence, as we overview later in \Cref{sec:overviewlinear}.

\subparagraph{Sublinear MPC Regime} We design the first deterministic sublogarithmic algorithm for finding a $2$-ruling set when the memory per machine is strictly sublinear. 

\begin{theorem}\label{thm:sublinear-memory}
    There is a deterministic sublinear \mpc\ algorithm that finds a $2$-ruling set in $O(\sqrt{\log \Delta} \cdot \log \log \Delta + \log \log^* n)$ rounds using $O(n^{1+\varepsilon}+m)$ global space, for any constant $\varepsilon > 0$. Moreover, the same algorithm runs in $O(\sqrt{\log \Delta} \cdot \log \log n)$ using global space $O(n+m)$.
\end{theorem}

For $\Delta \gg \log^* n$, our algorithm gives an almost quadratic improvement over the runtime obtained using the MIS algorithm of \cite{CDP24}, and gets closer to the $\tilde O(\log^{1/6} \Delta + \log \log n)$ randomized complexity of \cite{KPP20}. It is worth noting that it matches the conditionally-optimal runtime of $\Omega(\log \log n)$ when $\Delta = O(2^{\log^2 \log n / \log \log \log n})$, even though, being it not component-stable, the lower bound does not apply.

This algorithm (\Cref{sed:sublinear}) is obtained by derandomizing the sparsification developed by \cite{KP12} for solving $2$-ruling sets in the LOCAL model. Specifically, we show that a randomized $O(1)$-LOCAL downsampling step can be carried out in only $O(\log \log \Delta)$ rounds deterministically in \mpc\ with strongly sublinear space per machine and optimal global space. To achieve that, we combine several well-established derandomization tools such as limited independence, the method of conditional expectation, and coloring for reducing seed length, as we discuss in \Cref{overviewsublinear}.

We also note that our techniques may be more general and apply to $\beta$-ruling sets for $\beta > 2$. Concretely, one may combine our result with the framework of \cite{BKP14} to obtain faster \mpc\ $\beta$-ruling sets algorithms. This direction is left for future work.

\subsection{2-Ruling Sets: Technical Overview} 
We present the main intuition behind the recent constant-round randomized algorithm by \cite{CKPU23} in the linear regime of \mpc\ and the randomized $\tilde O(\sqrt{\log n})$-round LOCAL algorithm by \cite{KP12}, which is also closely followed by subsequent works \cite{HPS14, KPP20, PP22}. Then, we provide an overview of our deterministic algorithms and the main ideas that lead to randomness-efficient analyses.

\subsubsection{Linear Memory Regime}\label{sec:overviewlinear}
\subparagraph{Randomized Constant-Round Algorithm} The constant-round $2$-ruling set algorithm by \cite{CKPU23} relies on computing an MIS iteratively on subgraphs of linear size, which can be solved locally on a single machine. Their algorithm samples each vertex $v$ from $V$ and includes it in $V_{\text{samp}}$ independently with probability $1/\sqrt{\deg(v)}$. This sampling primitive is shown to give two useful structural properties, with high probability. First, the induced subgraph $G[V_{\text{samp}}]$ has a linear number of edges. Second, a certain MIS computation on $G[V_{\text{samp}}]$ returns an independent set that is at distance at most two from all but at most $n/\sqrt{d}$ vertices with degree $[d, 2d)$ in the original graph $G$, for each $d \in \{2^{\lfloor\log \Delta\rfloor}, 2^{\lfloor\log \Delta\rfloor - 1},\ldots, \Omega(1)$\}. Then, it is shown that, after two repetitions, the number of remaining edges for each degree class $d$ is at most $n/\poly(d)$, which sums up to $O(n)$ over all $d$'s.

Their analysis of the above sampling process relies on full independence in the sense that random decisions of any node influence its neighbors at distance at most three. Then, each node influences only up to $n^{3 \alpha}$ many nodes by assuming that any node has degree at most $n^{\alpha}$, for constant $\alpha > 0$. This property is exploited to union bound over large sets of independent nodes in $G^7$, since nodes at distance 8 are enough far apart not to influence one another. Clearly, this property breaks apart under our constraint of limited independence and requires to analyze the sampling process differently.

\subparagraph{Constant-Round Derandomization}
In a nutshell, we show that the same asymptotic guarantees as that provided by the above randomized algorithm can be achieved deterministically.
While it is easy to show that their initial sampling step gives a subgraph with a linear number of edges in expectation, even under pairwise independence, the main challenge is to prove that only $n/d^{\Omega(1)}$ nodes survive across all $O(\log \Delta)$ $d$-degree classes, simultaneously. Establishing the same polynomial decrease (in $d^{\Omega(1)}$) of the size of each $d$-degree class ensures the same constant-round complexity. 

Our key modification to \cite{CKPU23}'s analysis is to increase the threshold for a node to be called good. We say that a node of degree $d$ is good if it has at least $d^{\Omega(1)}$ neighbors in $G[V_{\text{samp}}]$, as opposed to the $\Theta(\log n)$ requirement of \cite{CKPU23}. This leads to the following two properties.
 
First, in the sampling step, we prove that each good node of degree $d$ is covered with probability $1 - 1/\poly(d)$ and that suffices. In fact, through the method of conditional expectation, non-covered nodes will induce at most $O(n)$ edges.

Second, in the MIS step, we prove that remaining ``bad'' nodes are at most $n/d^{\Omega(1)}$ for each degree class, without any assumption on the maximum degree. To achieve that, we combine a pairwise independent MIS algorithm (similar to that of \cite{FGG23}) with a pessimistic estimator that notably expresses the progress made over \textit{all} degree classes as a \textit{single} expectation. This expectation can then be obtained by means of standard derandomization tools.

\subsubsection{Strongly Sublinear Memory Regime}\label{overviewsublinear}
\subparagraph{Randomized $2$-Ruling Set Sparsification} 
The central step of the $2$-ruling set algorithms by \cite{KP11, KPP20} is a sparsification procedure that returns a subgraph $G'$ of sufficiently small maximum degree. Then, computing a maximal independent set on $G'$ has time proportional to its maximum degree, and yields a $2$-ruling set that covers all vertices in $G$ which have a neighbor in $G'$. 

They construct a subgraph $G'$ of maximum degree $O(f \cdot \log n)$ such that any (high-degree) node with a degree in $[\Delta, \Delta/f]$ in $G$ has a neighbor in $G'$, for some parameter $f \ge \log n$. It is easy to see that sampling each vertex $v \in V$ with probability $f \cdot \log n / \Delta$ independently ensures that every vertex with degree at least $\Delta/f$ will have a sampled vertex in its neighborhood with high probability. 

We just focused solely on covering vertices with degrees in $[\Delta, \Delta/f]$. It turns out that, by each time removing the subgraph $G'$ and its neighbors, the same sampling step can be repeated $O(\log_f \Delta)$ times, where in the $j$-th step nodes with degrees in $[f^{\log_f \Delta - (j-1)}, f^{\log_f \Delta - j}]$ are covered, with $j \in [\log_f \Delta]$. This simple process leads to a randomized round complexity of $O(\log f + \log_f \Delta + \poly\log\log n)$ by applying any MIS algorithm that runs in $O(\log \Delta + \poly\log\log n)$ rounds \cite{Gha16, GU19} on the union of all subgraphs, which have no conflicts by construction. Then, $f = 2^{\sqrt{\log \Delta}}$ is chosen to achieve a runtime of $O(\sqrt{\log \Delta} + \poly\log \log n)$. 

\subparagraph{Deterministic $2$-Ruling Set Sparsification} 
Our goal is to replace the above randomized sampling with a deterministic sampling that returns a subgraph $G'$ with the same properties as those returned by the above construction \cite{KP11, KPP20}. We slightly alter the sampling guarantees to allow for a relaxed maximum degree in $G'$ of up to $\poly(f)$ instead of $O(f \log n)$. Instead of sampling each vertex with probability $f \cdot \log n / \Delta$ randomly and independently in a single round, we sample them in a deterministic manner in $O(\log \log \Delta)$ rounds. The way in which we design this deterministic sampling step is explained next.

The standard approach is to limit the randomness by sampling vertices using a carefully selected $k$-wise independent hash function. A naive implementation that samples vertices with probability $\frac{\poly(f)}{\Delta}$ would need a family of $k$-wise independent hash functions with $k = \Omega(\log_f n)$, since each vertex has $\poly(f)$ expected sampled neighbors. The need for $\Omega(\log_f n)$-wise independence results in a seed of length $\Omega(\log_f n \cdot \log \Delta)$. Since in $O(1)$ MPC rounds only $O(\log n)$ bits can be fixed, this one-step process appears to require $\Omega(\frac{\log \Delta}{\log f})$ many rounds\footnote{Here, shortening the seed length using a family of $\eps$-approximate $k$-wise independent hash functions still requires $\omega(1)$ \mpc\ rounds.}, which is very far from being sublogarithmic.

Our approach to make this construction randomness-efficient relies on breaking down the sampling process into $O(\log \log \Delta)$ sub-sampling processes, each of which has weaker guarantees but requires only $O(1)$ rounds. 
In particular, the basis of our process is a simple, deterministic, constant-round routine that decreases the maximum degree by a $O(\sqrt{\Delta})$-factor, while ensuring that the maximum-to-minimum degree ratio of $O(f)$ is maintained, i.e., each vertex $v$ has degree roughly $|N_G(v)| / \sqrt{\Delta}$ in $G'$. 

Then, we repeatedly apply this degree-reduction routine to sparsify the neighborhoods of high-degree vertices until their degree drops to $2^{O(\log f)}$. It is easy to see that this requires at most $O(\log \log \Delta)$ repetitions. However, in each iteration, some downsampled neighborhoods may deviate from their expectation, say by an $\epsilon$-factor. Such deviation is amplified each time, resulting in a potential error of $\epsilon^{O(\log \log \Delta)}$. Nevertheless, through a suitable $f$ and $\epsilon$, we can minimize the error and show that the subgraph $G'$ has $\poly(f)$ maximum degree. Therefore, we can iterate through the $O(\log_f \Delta)$ degree classes (as in the randomized case) and apply our deterministic degree reduction to achieve the same result, up to a $O(\log \log \Delta)$ factor.

\subparagraph{Further Comparison} Several sparsifications for MIS and $2$-ruling sets in \local\ and low-memory \mpc\ have been studied. We include a brief comparison with the works of \cite{CDP21mis, MPU23, KPP20}.

A deterministic $O(1)$-round sampling process appeared in the MIS algorithm of \cite{CDP21mis}. There, the goal is to reduce the maximum degree to at most $n^\epsilon$ while ensuring that the resulting subgraph maintains enough edges and the distribution of degrees is still representative of the original graph. They decrease the maximum degree by an $n^{\Omega(1)}$-factor for $O(1)$ times, until the desired bound is achieved. Since the expected new maximum degree is still on the order of $n^{\Omega(1)}$, concentration around the expectation can be achieved with $O(1)$-wise independence, and thus derandomized in $O(1)$ rounds. In contrast, in $2$-ruling set, the main challenge is to subsample the neighborhoods of nodes with degree $d \ll n^{\Omega(1)}$. In fact, applying a similar subsampling method would require $\Omega(\log_d n)$-wise independence and $\Omega(\frac{\log \Delta}{\log f})$ rounds, as explained in the paragraph above. Thus, while the method in \cite{CDP21mis} is effective for high-degree nodes with $d = n^{\Omega(1)}$, handling smaller degrees requires a different approach.

The ruling set algorithm of \cite{MPU23} introduces a \congest\ sparsification that runs in $O(\log^2 n)$ rounds and deals with $O(\log \Delta)$ degree classes. There, a single sampling step requires a seed of length $O(\log^2 n)$ as they require guarantees stricter than ours. Specifically, their sparsification must maintain a low diameter and ensure proper coverage. Although their derandomization is \congest-efficient, it would require $O(\log n)$ \mpc\ rounds, making it unsuitable to our setting.

Finally, we note that the faster randomized $2$-ruling set algorithm of \cite{KPP20} relies on (informally) performing graph exponentiation on a sparsified subgraph. This approach relies on fixing the randomness of future iterations in advance, which simplifies the process of speeding up algorithms in  \local. The main challenge in adapting this approach to a deterministic setting is that existing techniques are generally effective at derandomizing only $O(1)$ steps of an algorithm. They do not easily extend to derandomize algorithms that simulate $\Omega(1)$ randomized rounds locally on each single machine via graph exponentiation. Consequently, achieving the same speed up deterministically appears to require a novel approach.
\section{Preliminaries}
In our analyses, we will use the notation $\poly(\cdot)$ to refer to $(\cdot)^c$, for a constant $c > 0$ at the exponent that can be made arbitrarily large without affecting  asymptotic bounds.

\subparagraph{Primitives in MPC}
    We recall that basic computations can be performed in the \mpc\ model with strongly sublinear local memory in $O(1)$ rounds deterministically \cite{G99, GSZ11}. 
    
    Therefore, tasks such as computing the degree of each vertex, ensuring neighborhoods of all vertices are stored on single machines, and collecting certain subgraphs onto a single machine will be used as black-box tools.

\subparagraph{Derandomization Framework}
A rich and successful line of research has studied the derandomization of algorithms in the parallel and distributed setting. In the \mpc\ model, classic derandomization schemes using limited independence and the method of conditional expectation \cite{Lub93, MNN94}, can be augmented with the power of local computation and global communication to achieve the expected result in $O(1)$ rounds.

We will often use the concepts of $k$-wise independence and family of $k$-wise independent hash functions (see, e.g., \cite{MR95, Rag88}). Given a randomized process that works under $k$-wise independence, it is known how to construct a $k$-wise independent family of hash functions.
    \begin{lemma}[\cite{ABI86, CG89, EGL+98}]\label{lem:k-wise}
        For every $N, k, \ell \in \mathbb{N}$, there is a family of $k$-wise independent hash functions $\mathcal{H} = \{h : [N] \rightarrow \{0, 1\}^{\ell}\}$ such that choosing a uniformly random function $h$ from $\mathcal{H}$ takes at most $k(\ell + \log N) + O(1)$ random bits, and evaluating a function from $\mathcal{H}$ takes time $\poly(\ell, \log N)$ time.
    \end{lemma}
Moreover, to show concentration around the expected value under $k$-wise independence, we will use the following tail bound.
\begin{lemma}[Lemma 2.3 of~\cite{BR94}]\label{lem:kwise_bound}
    Let $k \ge 4$ be an even integer. Let $X_1,\ldots,X_n$ be random variables taking values in $[0,1]$. Let $X = X_1 + \ldots + X_n$ denote their sum and let $\mu \le \E[X]$ satisfying $\mu \ge k$. Then, for any $\epsilon > 0$, we have
    \begin{equation*}
        \Pr\left[|X - \E[X]| \ge \epsilon \cdot \E[X]\right] \le 8 \left( \frac{2k}{\epsilon^2 \mu}\right)^{k/2}.
    \end{equation*}
\end{lemma}

We consider randomized algorithms that succeed in expectation when their random choices are made using a family of $k$-wise independent hash functions $\mathcal{H}$. Once our algorithm (randomly) picks a hash function $h$, then all choices are made deterministically according to $h$. Thus, our problem is that of deterministically finding a hash function that achieves a result as good as the expectation. 

The by-now standard \mpc\ derandomization process can be broken down into two parts: (i) show that the family of hash functions $\mathcal{H}$ has size $\poly(n)$ \textit{and} produces the desired result in expectation, and (ii) find one good hash function by applying the method of conditional expectation in a distributed fashion. We will focus on establishing (i), since (ii) can then be achieved by known \mpc\ derandomization methods introduced by earlier works \cite{CPS17, CC22, CDP21mis} to which we refer for further details.
It is worth mentioning that for step (ii) to be solved using earlier tools as a black-box, the aimed expectation should be expressed as a sum of locally computable quantities by each individual machine, i.e., the individual expectation of each node that a machine stores.

\section{Deterministic 2-Ruling Set in Linear MPC}\label{sec:linear}
We first introduce the reader to several sets of nodes that play a crucial role in our algorithm. These sets of nodes are defined to reflect how a node will be handled by our algorithm. Specifically, the core of the algorithm is a downsampling procedure that outputs a sufficiently small subgraph on which we will compute a maximal independent set with the goal of \textit{ruling} a large fraction of nodes in the original graph. 

Observe that if a node has a neighbor in the downsampled graph, then it will have some node in the maximal independent set at distance at most two. This means that if a node is likely to have a sampled neighbor, then it is likely to be ruled, and we call such a node \textit{good}. In the following, our definitions and algorithm are parameterized by a constant $\varepsilon = 1/40$, which has not been optimized.
\begin{definition}[Good Node]
    A node $v \in G$ is good if it satisfies $\sum_{u \in N(v)} \frac{1}{\sqrt{\deg (u)}} \ge \deg(v)^{\varepsilon}$.
\end{definition}
If a node $v$ is not good, i.e., $\sum_{u \in N(v)} \frac{1}{\sqrt{\deg (u)}} <\deg(v)^{\varepsilon}$, then we say that $v$ is a \textit{bad} node. Bad nodes are split into $O(\log \Delta)$ degree classes as follows. Let $d_0$ be a sufficiently large constant and $d_{\text{max}} = \ceil{\log \Delta}$. 
\begin{definition}[Bad Node Classes]
    For $d \in \{2^{d_0}, 2^{d_0 + 1}, \ldots, 2^{ d_{\text{max}}}]$, the set $B_d$ includes all bad nodes with degree in $[d, 2d)$. 
\end{definition}
Therefore, bad nodes are likely to have few sampled nodes. This fact motivates the following observation. If a (bad) node has \textit{many} bad nodes within its $2$-hop neighborhood, then it is likely that at least one of such bad ones is in the maximal independent set. If that is the case, we call such nodes \textit{lucky} bad nodes, as specified in the following definition.
\begin{definition}[Lucky Bad Nodes]
    For $d \in \{2^{d_0}, 2^{d_0 + 1}, \ldots, 2^{ d_{\text{max}}}]$, the set $\overline{B}_d \subseteq B_d$ includes each node $u \in B_d$ such that $u$ has a neighbor $w$ with $|N(w) \cap B_d| \geq 6d^{0.6}$. If there are multiple such $w$'s, pick one arbitrarily and let $S_u$ be an arbitrarily chosen subset of $N(w) \cap B_d$ such that $\card{S_u} = 6d^{0.6}$. 
\end{definition}

With these definitions in mind, we are now ready to present our deterministic constant-round $2$-Ruling Set algorithm in the linear regime of MPC. 

The algorithm operates in three simple steps: Sampling, Gathering, and MIS Computation. The first step of the algorithm samples each node $v$ with probability $\deg^{-1/2}(v)$. The sampling probability is chosen to ensure that the downsampled graph has a linear number of edges. Moreover, we will slightly alter the downsampled graph to include all nodes that do not satisfy certain requirements, without affecting the asymptotic size of this subgraph. Therefore, in the second step, we will be able to collect such subgraph onto a single machine. Then, the MIS computation begins by running one iteration of Luby's MIS on (part of) the subgraph from the previous step and continues by extending such independent set to a maximal one locally.

We will prove several desirable properties about the three-step algorithm above that lead to a reduction of a $d^{\Omega(1)}$-factor for each degree class $d$. Therefore, by repeating this three-step algorithm $O(1)$ times, the number of edges over all degree classes converges to $O(n)$ and thus can be collected and solved locally, completing the proof of \cref{thm:linear-memory}.

Next, we present the algorithm in more detail and then proceed to analyzing its three steps with a particular focus on randomness efficiency. In fact, such randomness-efficient analyses will allow for a simple derandomization. 

\subsection{The Algorithm}
\subparagraph{Sampling Step} Let $G = (V, E)$ be the input graph with $n$ vertices and $m$ edges. Let $V_{\text{samp}}$ denote the set of sampled vertices. We include each vertex $v \in V$ in $V_{\text{samp}}$ with probability $p_v = \frac{1}{\sqrt{\text{deg}(v)}}$, according to a family of $k$-wise independent random variables with $k=O(1)$.

\subparagraph{Gathering Step} We gather several subsets of nodes whose (combined) induced subgraph will be shown to have a linear number of edges. Gathered nodes are those either sampled in the previous step or not satisfying certain properties as formally defined below. Let $V^{*}$ denote the union of the following node subsets, which are being gathered locally onto a single machine: 
\begin{enumerate}
    \item The set of sampled nodes $V_{\text{samp}}$;
    \item Every good node that is \textit{not sampled} and has \textit{no sampled neighbors};
    \item For each $d$, every lucky bad node $u \in \overline{B_d}$ that has either less than $d^{0.1}$ sampled nodes in $S_u$ or one of the sampled nodes in $S_u$ has more than $d^{2\varepsilon}$ sampled neighbors; as formalized in \cref{lem:bad_nodes_not_ruled}.
\end{enumerate}

\subparagraph{MIS Computation} Our goal is now to compute a maximal independent set on the locally gathered subgraph $G[V^*]$ to rule all but roughly at most a $\Delta^{\Omega(1)}$-fraction of nodes in $G$. We achieve this by first computing a partial MIS on the sampled bad vertices, i.e., $\bigcup_{d}B_d \cap V_{\text{samp}}$, using a variation of Luby's algorithm as detailed in the proof of \cref{lem:bad_nodes_mis}. Afterward, we can simply compute an MIS locally (and thus sequentially) on the remaining vertices, which are not incident to the partial MIS computed earlier.

\subparagraph{Output Properties} We expect that the output given by the \textit{derandomization} of the above three-step process satisfies the following properties. We will later use these properties to achieve a deterministic constant-round complexity. Observe that we can ignore constant-degree nodes since they can be gathered and dealt with locally at last.
\begin{itemize}
    \item \textbf{Good nodes:} All good nodes in $G$ are ruled after the MIS step.
    \item \textbf{Uncovered lucky bad nodes:} For each $d$, after the computation of a partial MIS, only a $d^{\Omega(1)}$-fraction of lucky bad nodes remains uncovered.
    \item \textbf{Uncovered bad nodes:} For each $d$, the number of bad nodes in $B_d \setminus \overline{B}_d$ is only a $d^{\Omega(1)}$-fraction of all nodes with initial degree at least $d$ in $G$ .
\end{itemize}

\subsection{Analysis}

We first establish that good nodes are likely to have a neighbor in $V_{samp}$. Since we will compute an MIS on $V^* \supseteq V_{samp}$, such good nodes will be at distance at most $2$ from a node in the MIS. Moreover, good nodes that have no sampled neighbor will be shown to be incident to a linear number of edges, allowing us to gather them as part of $V^*$.
\begin{lemma}\label{lem:good_nodes_not_ruled}
    Every good vertex $v$ has a neighbor in $V_{\text{samp}}$ with probability at least $1 - \frac{1}{\text{poly}(\text{deg}(v))}$.
\end{lemma}
\begin{proof}
For any vertex $u$, let $X_u$ be the indicator random variable for the event $u \in V_{\text{samp}}$, and $X$ be the random number of neighbors of $v$ in $V_{\text{samp}}$. Further, let $\mu := \mathbb{E}[X] = \sum_{u \in N(v)}\mathbb{E}[X_u] = \sum_{u \in N(v)}\Pr[X_u = 1] \ge \text{deg}(v)^{\varepsilon} \gg k$, since nodes of constant degree can be ignored and dealt with separately at last by collecting them onto a single machine. By applying \Cref{lem:kwise_bound}, we have
\begin{align*}
    \Pr[X=0] \le \Pr[|X - \mu| \ge \mu] &\le 8\cdot \left(\frac{k\mu + k^2}{\mu^2}\right)^{k/2} \le 8\cdot \left(\frac{2k}{\mu}\right)^{k/2} = \frac{1}{\text{poly}(\text{deg}(v))},
\end{align*}
which proves the lemma.
\end{proof}

Toward the goal of ruling lucky bad nodes, we next show that bad nodes are likely to have few sampled neighbors. This means that sampled bad nodes, by having a low degree in the sampled graph, will have higher chances of being in the partial MIS computed later.

\begin{lemma}\label{lem:sampled_bad_nodes_neighbors}
    Any node $u \in B_d$ has at most $d^{2 \varepsilon}$ sampled neighbors with probability at least $1 - \frac{1}{\text{poly}(d)}$.
\end{lemma}
\begin{proof}
Recall that for any $u \in B_d$, it holds that $\sum_{w \in N(u)}\frac{1}{\sqrt{\text{deg}(w)}} < \text{deg}(u)^{\varepsilon}$. We will use this fact to prove that the number of sampled neighbors does not deviate by more than $O(d^{2 \varepsilon})$ with probability at least $1 - \frac{1}{\text{poly}(d)}$. 
Let $X_w$ be the indicator random variable for the event $w \in V_{\text{samp}}$, and $X$ be the random number of neighbors of $u$ in $V_{\text{samp}}$. Let $\mu = \mathbb{E}[X] = \sum_{w \in N(u)}\mathbb{E}[X_w] = \sum_{w \in N(u)}\Pr[X_w = 1] < \text{deg}(u)^{\varepsilon} < 2d^{\varepsilon}$. By applying \Cref{lem:kwise_bound}, we get
\begin{align*}
    \Pr[|X - \mu| \ge d^{2 \varepsilon} - \mu] &\le 8\cdot \left(\frac{k^2 + k\mu}{(d^{2 \varepsilon} - \mu)^2}\right)^{k/2} \le 8\cdot \left(\frac{2k^2}{d^{ \varepsilon}}\right)^{k/2} = \frac{1}{\text{poly}(d)}.
\end{align*}
Note that for small values of $d$, our constant $d_0$ can be chosen such that $2^{d_0 \cdot \varepsilon} = \Omega(k^2)$.
\end{proof}

The next lemma proves that each lucky bad node $u$ has a large number of nodes sampled out of its set $S_u$. Specifically, we need to show that the number of sampled nodes in $S_u$ is higher than the degree of such nodes in the sampled graph. This fact will be used to ensure that lucky bad nodes have a vertex, within their $2$-hop neighborhoods, in the MIS, thereby, ensuring their coverage.

\begin{lemma}\label{lem:bad_nodes_not_ruled}
    For any lucky bad node $u$, its set $S_u \subseteq B_d$ of cardinality $6d^{0.6}$ contains at least $d^{0.1}$ sampled nodes and each sampled node in $S_u$ has at most $d^{2\varepsilon}$ sampled neighbors with probability at least $1 - \frac{1}{\text{poly}(d)}$.
\end{lemma}
\begin{proof}
By \cref{lem:sampled_bad_nodes_neighbors} and a union bound over the set $S_u$ of $6d^{0.6}$ nodes, none of them has more than $d^{2 \varepsilon}$ sampled neighbors with probability at least $1 - \frac{1}{\text{poly}(d)}$. Our goal is now to prove that the number of sampled vertices within $S_u$ is less than $d^{0.1}$ with probability at most $\frac{1}{\text{poly}(\text{deg}(u))} = \frac{1}{\text{poly}(d)}$. 

Let $X$ be the random number of sampled vertices in $S_u$, and let $\mu = \mathbb{E}[X] \ge 3d^{0.1}$, since each vertex in $B_d$ is sampled with probability at least $1/\sqrt{2d}$. 
By applying \Cref{lem:kwise_bound}, the probability of $X$ deviating by more than $d^{0.1}$ from its expected value is
\begin{align*}
    \Pr[|X - \mu| \ge \mu - d^{0.1}] &\le 8\cdot \left(\frac{2k\mu}{(\mu - d^{0.1})^2}\right)^{k/2} \le 8\cdot \left(\frac{2k}{d^{0.1}}\right)^{k/2} = \frac{1}{\text{poly}(\text{deg}(u))}.
\end{align*}
\end{proof}
 
We now use the above lemmas, together with a bound on the number of edges induced by the sampling step, to prove that our gathering step effectively collects $O(n)$ edges.

\begin{lemma}\label{lem:first_sampling}
The subgraph induced by $G[V^*]$ has $O(n)$ edges in expectation. 
\end{lemma}
\begin{proof}
    Our goal is to prove that the expected sum of the original degrees of nodes in $V^*$ is $O(n)$, which clearly upper bounds the number of edges in the induced subgraph. To do so, we analyze each subset individually.
    
    We first analyze the expected number of edges induced by $V_{samp}$. Let $X$ denote the random number of edges within the subgraph $G[V_{\text{samp}}]$. Let $Y_e$ be an indicator random variable for the event that edge $e$ is in $G[V_{\text{samp}}]$. To aid our analysis, we orient each edge in the graph from the endpoint with lower degree to the endpoint with higher degree. Now, consider an edge $e = (u, v)$ with $\text{deg}(u) \le \text{deg}(v)$. Vertices $u$ and $v$ are each sampled with probability at most $\frac{1}{\sqrt{\text{deg}(u)}}$. By pairwise independence, the probability of edge $e$ being in $G[V_{\text{samp}}]$ is bounded by $\frac{1}{\text{deg}(u)}$. Consequently, the expected number of edges is $\mathbb{E}[X] = \sum_{v \in V}\sum_{e \in \text{out}(v)}\mathbb{E}[Y_e] \le \sum_{v \in V}\sum_{e \in \text{out}(v)} \frac{1}{\text{deg}(u)} = O(n)$.

    Next, let $\overline{V}_{\text{good}}$ denote the set of good nodes that have no sampled neighbor and $Y$ the random number of edges incident to $\overline{V}_{\text{good}}$ in $G$. By \cref{lem:good_nodes_not_ruled}, each good node $v$ is in $\overline{V}_{\text{good}}$ with probability at most $1/\text{poly}(\text{deg}(v))$. Thus, 
        \begin{equation*}
            \E[Y] \le \sum_{v \in V} \deg(v) \cdot \Pr[v \in \overline{V}_{\text{good}}] \le \sum_{v \in V} \frac{\deg(v)}{\text{poly}(\text{deg}(v))} = O(n).
        \end{equation*}

    Finally, let the set $B_d' \subseteq \overline{B_d}$ include each unlucky bad node $u$ such that either less than $d^{0.1}$ vertices in $S_u$ are sampled or any sampled node in $S_u$ has more than $2d^{\varepsilon}$ sampled neighbors. By \cref{lem:bad_nodes_not_ruled}, each node $u$ is in ${B}_d'$ with probability at most $1/\text{poly}(d)$. Let $Z$ be the random number of edges incident to ${B}_d'$. We have
        \begin{equation*}
            \E[Z] \le \sum_{i = d_0}^{d_{\text{max}}} \sum_{u \in B_{2^i}} \deg(u) \cdot \Pr[u \in {B}_{2^i}'] \le \sum_{i = d_0}^{d_{\text{max}}} \sum_{u \in B_{2^i}} \frac{2d}{\text{poly}(d)} \le \sum_{i = d_0}^{d_{\text{max}}} |B_{2^i}| = O(n).
        \end{equation*}
\end{proof}

\subparagraph{Derandomize Sampling and Gathering Steps} We are now ready to discuss how the above Sampling and Gathering steps can be turned into a deterministic linear \mpc\ algorithm. Recall that each vertex is sampled according to a family of $k$-wise independent random variables with $k=O(1)$. A family $\mathcal{H}$ of $k$-wise independent hash functions such that $h \in \mathcal{H} : [n] \rightarrow [n^3]$ can be specified using a random seed of length $O(\log n)$, meaning that $\card{\mathcal{H}} = \poly(n)$. Each $h$ maps the $n$ vertex IDs (assumed to be from $1$ up to $n$) to an integer in $[n^3]$. Then, each vertex is sampled and belongs to $V_{\text{samp}}$ iff its ID is mapped to an integer that is at most $\floor{n^3 / \sqrt{deg(v)}}$ with respect to $h$, where the floor affects results only asymptotically. Each vertex can now locally check whether it will be included in $V^*$ for a specified hash function $h$. In fact, the machine that $v$ is assigned to stores all $v$'s neighbors and the set $S_v$ if $v$ is a lucky bad node. Therefore, it is easy to see that each node can computed the objective function $|E(G[V^*])|$ locally, and we can thus apply the distributed method of conditional expectation. Since $\card{\mathcal{H}} = \poly(n)$, after a constant number of rounds we will find a $h$ that ensures $|E(G[V^*])| = O(n)$.\\

We now turn to analyzing the MIS step. Recall that we first compute a partial MIS on the sampled bad nodes in order to rule all but a small fraction of lucky bad nodes. The next lemma explains how such an independent set is being computed.

\begin{lemma}\label{lem:bad_nodes_mis}
     Let $\hat{B_d}$ include each node $u \in \overline{B_d}$ that satisfies the property of \cref{lem:bad_nodes_not_ruled}. After the partial MIS computation, each node $u \in \hat{B_d}$ will be ruled with probability at least $1 - \frac{45}{d^{\varepsilon}}$ for all $d \in [d_0, d_{\text{max}}]$. This result depends only on the randomness used in the MIS computation.
\end{lemma}
\begin{proof}
    We analyze one step of (a variation of) Luby's algorithm that builds an independent set $\mathcal{I}$ on the set of sampled bad vertices $\bigcup_{d}B_d \cap V_{\text{samp}}$. 
    We will fix a seed specifying a hash function from a pairwise independent family $\mathcal{H}$. Let $v \in (\bigcup_{d}B_d \cap V_{\text{samp}})$. An hash function $h$ maps node $v$ to a value $z_v \in [n^3]$. Then, $v$ joins the independent set $\mathcal{I}$ iff $z_v < z_w$ for all $w \sim v$ \textit{and} $z_v < \frac{n^3}{d^{3\varepsilon}}$, where $w \in N(v) \cap (\bigcup_{d}B_d \cap V_{\text{samp}})$.

    By \cref{lem:bad_nodes_not_ruled}, each node $u \in \hat{B_d}$ has at least $d^{0.1}$ nodes from $S_u$ that are sampled, each of which has at most $d^{2\varepsilon}$ sampled neighbors. For the purpose of the analysis, let the set $A_u$ include exactly $d^{0.1} = d^{4\varepsilon}$ of such nodes and let $\{ X_v \}_{v \in A_u}$ be the random variables denoting the event that $v$ joins $\mathcal{I}$. We denote $X = \sum_{v \in A_u} X_{v}$ as their sum. For any $v$, we have
    \begin{align*}
        \frac{1}{d^{3\varepsilon}} - \frac{1}{n^3} \le \Pr\left[z_v < \frac{n^3}{d^{3\varepsilon}}\right] \le \frac{1}{d^{3\varepsilon}}.
    \end{align*}
    By pairwise independence, 
    \begin{align*}
        \Pr[X_v = 1] &\ge \Pr\left[z_v < \frac{n^3}{d^{3\varepsilon}}\right] - \sum_{\substack{v' \in N(v) \cap S(B)}}\Pr\left[z_{v'} \le z_v < \frac{n^3}{d^{3\varepsilon}}\right] \ge \frac{1}{d^{3\varepsilon}} - \frac{1}{n^3} - \frac{d^{2\varepsilon}}{d^{6\varepsilon}} \ge \frac{1}{3d^{3\varepsilon}}.
    \end{align*} 
    It follows that $\E[X] = \sum_{v \in A_u} \Pr[X_v = 1] \ge \frac{d^{\varepsilon}}{3}.$
    Our goal is now to bound $\Pr\left[X = 0\right]$. Observe that for any two vertices $v, v' \in A_u$, we have that $$\E[X_v X_{v'}] \le \Pr\left[z_v < \frac{n^3}{d^{3\varepsilon}} \, \cap \, z_{v'} < \frac{n^3}{d^{3\varepsilon}}\right] \le d^{-6\varepsilon},$$ by pairwise independence. 
    Thus, we get
    \begin{align*}
        \frac{\Var[X]}{\E[X]^2} &\le \frac{\sum_{v \in A_u} \Var[X_v] + \sum_{v, v' \in A_u} \Cov[X_v, X_{v'}]}{\E[X]^2}.
    \end{align*}
    We know that
    \begin{align*}
        \sum_{v \in A_u} \Var[X_v] &\le d^{4\varepsilon} \cdot \Pr[X_v = 1] (1 - \Pr[X_v = 1]) \le d^{4\varepsilon} \cdot \frac{1}{3d^{3\varepsilon}} = \frac{d^{\varepsilon}}{3}, \\
        \sum_{v, v' \in A_u} \Cov[X_v, X_{v'}] &\le d^{8\varepsilon} (\E[X_v X_{v'}] - \E[X_v]\E[X_{v'}]) \le d^{8\varepsilon} (d^{-6\varepsilon} - 1/9d^{6\varepsilon}) \le d^{2\varepsilon}.
    \end{align*}
    Therefore,
    \begin{align*}
        \Var[X] &\le \frac{d^{\varepsilon}}{3} + d^{2\varepsilon} \le \frac{4d^{\varepsilon}}{3}, \text{ and }
        \frac{\Var[X]}{\E[X]^2} \le \frac{\frac{4d^{\varepsilon}}{3}}{\left(\frac{d^{\varepsilon}}{3}\right)^2} = \frac{4d^{\varepsilon}}{3} \cdot \frac{9}{d^{2\varepsilon}} = \frac{36}{d^{\varepsilon}}.
    \end{align*}
    Applying Chebyshev's inequality, we have
    \begin{align*}
        \Pr[X = 0] \le \Pr\left[|X - \E[X]| \ge \E[X]\right] \le \frac{\Var[X]}{\E[X]^2} \le \frac{45}{d^{\varepsilon}}.
    \end{align*}
\end{proof}

The above lemma turns out not to be sufficient to derandomize our MIS step. In fact, we need to show that all degree classes of lucky bad nodes have a high enough chance of being ruled \textit{simultaneously}. This is due to the fact that in the derandomization process, we can control only \textit{one} objective function and not $O(\log \Delta)$ as the number of degree classes would appear to require. In the next lemma, we show how to define a pessimistic estimator that solves this issue.

\begin{lemma}\label{lem:bad_nodes_mis_grouped}
    After the partial MIS computation, all but at most $\frac{|\overline{B}_d|}{d^{\Omega(1)}}$ nodes will be ruled in expectation, for all $d$ simultaneously.
\end{lemma}
\begin{proof}
    Let us first reason about a fixed $d$ and then about all $d$'s simultaneously. 
    
    Recall that $\hat{B_d}$ include each node $u \in \overline{B_d}$ that satisfies the property of \cref{lem:bad_nodes_not_ruled}. There are at most $\frac{|\overline{B}_d|}{\poly(d)}$ vertices in $\overline{B}_d \setminus \hat{B_d}$ by \cref{lem:bad_nodes_not_ruled}. Then, any vertex in $\hat{B_d}$ is ruled with probability at least $1 - \frac{45}{d^{\varepsilon}}$ by \cref{lem:bad_nodes_mis}. Therefore, by linearity of expectation, the number of non-ruled vertices in is at most $45|\overline{B}_d| / d^{\varepsilon}$. 
    
    Our goal is now to define a \textit{single} objective function whose expected value ensures that the same asymptotic result holds for \textit{all} $d$ simultaneously. Let $X_d$ be the random number of unruled nodes in $\overline{B}_d$, for each $d$. We define our objective function $Q$, which will serve as a ``pessimistic estimator'', as a weighted sum of the $X_d$'s as follows.
    \begin{align*}
        Q = \sum_{i = d_0}^{d_{\text{max}}} X_{2^i} \cdot \frac{2^{i \cdot {\frac{\varepsilon}{2}}}}{|\overline{B}_{2^i}|},
    \end{align*}
    so that we get
    \begin{align*}
        \E[Q] &= \sum_{i = d_0}^{d_{\text{max}}} \E[X_{2^i}] \cdot \frac{2^{i \cdot {\frac{\varepsilon}{2}}}}{|\overline{B}_{2^i}|} \le \sum_{i = d_0}^{d_{\text{max}}} \frac{45|\overline{B}_{2^i}|}{2^{i \varepsilon}} \cdot \frac{2^{i \cdot {\frac{\varepsilon}{2}}}}{|\overline{B}_{2^i}|} = \sum_{i = d_0}^{d_{\text{max}}} \frac{45}{2^{i\varepsilon/2}} = O(1),
    \end{align*}
    where the convergency follows from choosing a sufficiently large constant $d_0 = O(\varepsilon^{-1})$. Observe that the expected value of $Q$ ensures that, for each set $\overline{B}_d$, the number of nodes which are not ruled after running our Luby's step is $X_d \le \E[Q] \cdot \frac{|\overline{B}_{d}|}{d^{\varepsilon/2}} = \frac{|\overline{B}_d|}{d^{\Omega(1)}}$.
\end{proof}

\subparagraph{Deterministic MIS Step} We now present an efficient derandomization of the above partial MIS computation in the linear \mpc\ regime. As discussed in \cref{lem:bad_nodes_mis}, our family $\mathcal{H}$ of pairwise independent hash functions has size $\card{\mathcal{H}} = \poly(n)$. Note that each lucky bad node $u$ can store in a single machine its set $S_u$ and all of their sampled neighbors since $\card{S_u} \cdot d^{2 \varepsilon} = O(d) = O(deg(u))$. Then, each vertex $u$ can check whether it will be ruled under a specified hash function $h$. Therefore, we can compute $u$'s contribution to $Q(h)$ locally, where $Q(h)$ is the objective function of \cref{lem:bad_nodes_mis_grouped} under a specified hash function $h$. This allows us to apply the distributed method of conditional expectation with objective $Q$ to find a good hash function with $Q(h) = O(1)$ in a constant number of rounds.

\subparagraph{Counting the bad nodes} Let $V_{\ge d}$ denote the set of all nodes in $G$ with initial degree at least $d$, and let the set $B_d^* \eqdef B_d \setminus \overline{B}_d$. It remains to prove that the set $B_d^*$ contains only a small fraction of nodes. The next lemma is equivalent to Lemma~9 of \cite{CKPU23} up to some parameters change.

\begin{lemma}\label{lem:bad_nodes_counting}
For any degree $d \in[2^{d_0}, 2^{d_{\text{max}}}]$, we have that $|{B}_d^*| \le 12|V_{\ge d}|/d^{0.4}$.
\end{lemma}

\begin{proof}
For a bad node $v$, it is easy to see by contradiction that at least $d/2$ of $v$'s neighbors have degree at least $d^{2(1 - \varepsilon)}/4$ (see also Lemma~8 of \cite{CKPU23}). Let $d' = \frac{d^{2(1 - \varepsilon)}}{4}$. Therefore, any node $v \in B_d^*$ has at least $d/2$ neighbors in $V_{\ge d'}$. Furthermore, any node in $V_{\ge d'}$ neighboring a node in $B_d^*$ has at most $6d^{0.6}$ edges connecting to nodes in $B_d \supseteq B_d^*$. As a result of these observations, we derive the following inequality:
\begin{align*}
    d/2 \cdot |{B}_d^*| \le 6|V_{\ge d'}|\cdot d^{0.6},
\end{align*}
which together with the fact that $d' \ge d$, for $d$ large enough, proves the lemma.
\end{proof}

\subparagraph{Bounding Total Runtime} In the above paragraphs, we showed how to achieve deterministically the properties required by our three-step algorithm outlined at the beginning of this section. We now rove that repeating this process $O(1)$ times reduces the size of the graph to $O(n/\Delta)$, implying that the remaining nodes can be collected and solved for locally.

\begin{lemma}\label{lem:global_progress}
At the end of the first iteration, the number of remaining uncovered vertices with degree at least $d$, denoted by $V_{\ge d}^{(1)}$, satisfies
\begin{equation*}
    |V_{\ge d}^{(1)}| \le |V_{\ge d}|/ d^{\varepsilon'}.
\end{equation*}
\end{lemma}
\begin{proof}
The remaining uncovered vertices are only bad nodes. An uncovered bad node of degree $[d, 2d)$ can be either in  ${B}_d^*$ (\cref{lem:bad_nodes_counting}) or remained uncovered after running the deterministic MIS step (\cref{lem:bad_nodes_mis_grouped}). Over all $d, \ldots, 2^{d_{\text{max}}}$, this leads to:
\begin{align*}
    |V_{\ge d}^{(1)}| \le \sum_{i = \log d}^{d_{\text{max}}} |{B}_{2^i}^*| + \frac{|\overline{B}_d|}{2^{\Omega(i)}} \le \sum_{i = \log d}^{d_{\text{max}}} \frac{12|V_{\ge 2^i}|}{2^{0.4\cdot i}} + \frac{|\overline{B}_d|}{2^{\Omega(i)}} \le |V_{\ge d}| \sum_{i = \log d}^{d_{\text{max}}} \frac{1}{2^{\Omega(i)}} = \frac{|V_{\ge d}|}{d^{\Omega(1)}},
\end{align*}
where the last inequality follows from  $|\overline{B}_d| \le |V_{\ge d}|$, and the final bound is due to the geometric sum being asymptotically dominated by the first term.
\end{proof}

Having established, in \cref{lem:global_progress}, the progress made at each iteration by our three-step process, we can now apply a simple induction to show the desired bound on the progress made after several iterations.

\begin{lemma}
After $O(1)$ iterations, the graph induced by uncovered nodes has $O(n)$ edges.
\end{lemma}

\begin{proof}
Let $V_{\ge d}^{(k)}$ denote the number of remaining uncovered vertices with degree at least $d$ at iteration $k$. Our goal is to prove that after $k$ iterations, it holds that $V_{\ge d}^{(k)} \le V_{\ge d}/d^{k\varepsilon'}$ so that for $k = O(1/\varepsilon')$, we get $V_{\ge d}^{(k)} \le V_{\ge d}/d^{1.1}$. The base case for $k=1$ follows from \cref{lem:global_progress}. Now, let us assume that $V_{\ge d}^{(k-1)} \le V_{\ge d}/d^{(k-1)\varepsilon'}$. By a straightforward application of \cref{lem:global_progress}, we have that $V_{\ge d}^{(k)} \le |V_{\ge d}^{(k)}|/ d^{\varepsilon'} \le V_{\ge d}/d^{k\varepsilon'}$, as desired. Now, since the number of nodes with degree $[d, 2d)$ is upper bounded by $|V_{\ge d}|$, the total number of edges is at most $\sum_{i = \log d_0}^{\log d_{\text{max}}} V_{\ge d} \cdot 2^{i+1 - 1.1\cdot i} = \sum_{i = \log d_0}^{\log d_{\text{max}}} O(n/2^{0.1\cdot i}) = O(n)$.
\end{proof}

\section{Deterministic 2-Ruling Set in Sublinear MPC}\label{sed:sublinear}

In this section, we show that for an input graph with maximum degree $\Delta$, a $2$-ruling set can be computed deterministically in the strongly sublinear memory regime of MPC in $\tilde O(\log^{1/2} n)$ rounds.

We start by introducing a simple, deterministic, constant-round routine that reduces the size of each high-degree neighborhood by a $\sqrt{\Delta}$-factor, where high-degree refers to node with degree at least $\log(n) \cdot \Delta^{0.6}$. For ease of exposition, assume that high-degree vertices form a set $U$, and that $V$ is the set of all vertices (including high-degree vertices) that are being downsampled. Therefore, we reason about a bipartite graph $G = (U \sqcup V, E)$, where each node in $u \in U$ is connected to each vertex $v \in N_G(u)$ in the other part. Our goal is to ensure that each vertex $u$ has roughly $N_G(u)/\sqrt{\Delta}$ neighbors deterministically. For simplicity, in the next lemma, we make two assumptions: (i) the neighbors of each vertex fit into a single machine, and defer the other case to \cref{lem:one-sampling-high-degr}; (ii) we are given a certain coloring of $G$ that we discuss how to achieve at the end of this section.
 
\begin{lemma}\label{lem:bi-partite-small-degree}
    Let $G$ be a graph with bipartition $V(G) = U \sqcup V$ and $\Delta$ be an upper bound on the maximum degree of any node in $U$ such that $\Delta \in O(n^\alpha)$ for some $\alpha < 1$. 
    Furthermore, assume that each node in $V$ is given a color out of a palette of $O(\Delta^6)$ colors, such that any two distinct nodes $v,v' \in V$ that have a common neighbor in $U$ are assigned distinct colors. 
    Then, there exists a deterministic constant-round sublinear MPC algorithm that computes a subset $V^{sub} \subseteq V$ such that for any node $u \in U$ with $deg_G(u) \geq \log(n) \cdot \Delta^{0.6}$, it holds that $|N_G(u) \cap V^{sub}| \in \left[ \frac{1}{3 \sqrt{\Delta}}|N_G(u)|,\frac{1}{\sqrt{\Delta}}|N_G(u)|\right]$. The global space usage is linear in the input size.
\end{lemma}

\begin{proof}
    Let us assume that each node $v \in V$ knows its own color $c_v$ of a coloring satisfying the above properties. Then, nodes in $V$ apply a hash function $h$ from a $k$-wise independent family $\mathcal{H}$ that maps each color to an integer in $[\lceil 3\sqrt{\Delta}/2 \rceil]$. A node $v$ is then sampled under $h$ iff $h(v) = 1$, which occurs with probability $1/\lceil 3\sqrt{\Delta}/2\rceil$, where the ceil affects our results only asymptotically. We choose $k = 4 c \log_{\Delta} n$, for constant $c > 0$, so that the seed length to select a hash function from $\mathcal{H}$ is at most $\ell = O(\log_{\Delta}n) \cdot \max\{O(\log \Delta^6), O(\log \sqrt{\Delta})\} = O(\log n)$, i.e., the family $\mathcal{H}$ has size $\poly(n)$. 
    
    We prove that for each vertex $u \in U$ with degree larger than $\log n \cdot \Delta^{0.6}$, the probability of having between $\frac{1}{3\sqrt{\Delta}}|N(u)|$ and $|N(u)|/\sqrt{\Delta}$ neighbors within $V^{sub}$ is at least $1 - \frac{1}{n^c}$, i.e., the count of $v$'s neighbors in $V^{sub}$ deviates by at most $\frac{1}{3\sqrt{\Delta}}|N(u)|$.
    For each neighbor $ v $ of $ u $, let $ X_v $ be an indicator random variable for the event $ v \in V^{sub} $. Define $ X = \sum_{v \in N(u)} X_v $ as the number of neighbors of $ u $ in $ V^{sub} $. Then, $ \mu  = \mathbb{E}[X] = \frac{2|N(u)|}{3\sqrt{\Delta}} \ge c\log n \Delta^{0.1}.$
    By applying \Cref{lem:kwise_bound}, we have:
    \begin{align*}
        \Pr[|X - \mu| \ge \mu/2] &\le 8\left(\frac{4k\mu + 4k^2}{\mu^2}\right)^{k/2} \le 8\left(\frac{16c^2\Delta^{0.1}\log^2 n + 32c^2\log^2 n}{\Delta^{0.2}c^2\log^2 n}\right)^{k/2}\\
        &\le 8\left(\frac{1}{\Delta^{0.1}}\right)^{\frac{4c}{2}\cdot\frac{\log n}{\log \Delta}} \le \frac{1}{n^{2c}}.
    \end{align*}
    Therefore, the expected number of high-degree vertices in $U$ whose count of sampled neighbors deviates by more than $\mu/2$ is at most $n^{2c - 1} < 1$. This means that we can apply the method of conditional expectation in a distributed fashion with as objective function the number of bad nodes, i.e., those whose sampled neighborhood deviates from the expectation by more than half. Since the memory capacity of each machine is $O(n^{\alpha})$, each machine can compute locally the contribution to the objective of all the vertices (and their neighbors) it stores. Therefore, after $O(1)$ rounds, we find a hash function such that \textit{all high-degree vertices} in $U$ have the desired number of sampled neighbors.
\end{proof}

Next, we discuss how to extend \cref{lem:bi-partite-small-degree} to handle the case in which not all neighbors of a vertex in $U$ can be collected onto a single machine. In particular, if $\Delta \gg n^{\alpha}$, then aiming for a reduction of a $\sqrt{\Delta}$-factor might not be viable, given the constrained local memory. Due to that, we slightly relax our goal and reduce our high-degree neighborhoods by a $n^{\eps}$-factor, for some constant $\eps < \alpha$. To achieve that, we split edges into groups so that each machine is assigned $n^{c\cdot\eps}$ edges, for $c > 1$. While we can only control the deviation of each single group of edges, we will be able to bound the overall number of neighbors, i.e., edges per node, using the fact that there are at most $\Delta / n^{c\cdot\eps}$ groups.

\begin{lemma}\label{lem:one-sampling-high-degr}
Let $G$ be a graph with bipartition $V(G) = U \sqcup V$. Let $\Delta$ be an upper bound on the maximum degree of any node in $U$ such that $\Delta \geq n^{10\eps}$, for some constant $\eps > 0$. Then, there exists a deterministic constant-round sublinear MPC algorithm that computes a subset $V^{sub} \subseteq V$ such that for any node $u \in U$ with $deg_G(u) \geq \log(n) \cdot \Delta^{0.6}$, it holds that $|N_G(u) \cap V^{sub}| \in \left[ \frac{1}{2 n^{\eps}}|N_G(u)|,\frac{3}{2n^{\eps}}|N_G(u)|\right]$. The global space usage is linear in the input size.
\end{lemma}
\begin{proof}
Consider an arbitrary vertex $u \in U$ with degree at least $\log(n) \cdot \Delta^{0.6}$. The idea is to split edges of $u$ into groups of size at most $n^{4\eps}$, which fits into the memory of one machine. Specifically, each machine holds $n^{4\eps}$ edges except for a single machine that holds any remaining edges, which are at most $n^{4\eps}$. Then, we sample nodes in $V$ with probability $n^{-\eps}$ according to a family of $O(1)$-wise independent hash function. Using a calculation similar to that of \cref{lem:bi-partite-small-degree}, we can find a hash function such that all groups of $n^{4\eps}$ edges have $n^{3\eps} \pm n^{2\eps}$ sampled edges. Then, the total number of sampled neighbors is at least
\begin{align*}
    \sum_{\text{machine } i}n^{3\eps} - n^{2\eps} \ge \floor{\frac{|N_G(u)|}{n^{4\eps}}} \cdot \left(n^{3\eps} - n^{2\eps}\right) \ge  \frac{|N_G(u)|}{n^{\eps}} - \frac{|N_G(u)|}{n^{2\eps}} - n^{3\eps} \ge \frac{|N_G(u)|}{2n^{\eps}},
\end{align*}
where $n^{3\eps} = o(\frac{|N_G(u)|}{2n^{\eps}})$ since $N_G(u) \ge n^{6\eps}$. An analogous calculation shows that the total number of sampled neighbors for any vertex $u$ is at most $\frac{3|N_G(u)|}{2n^{\eps}}$.
\end{proof}

We are now ready to present our $O(\log \log \Delta)$ sparsification. We show that we can find a subset of nodes incident to all nodes in $U$ such that their induced maximum degree is $2^{O(\log f)}$ for $f = 2^{\sqrt{\log \Delta}}$. This is achieved by repeating the sampling processes of Lemmas \ref{lem:bi-partite-small-degree} and \ref{lem:one-sampling-high-degr} for  $O(\log \log \Delta)$ times. Here, one key observation to bound the deviation is that in each run of \cref{lem:bi-partite-small-degree} only the lower tail may deviate up to a $1/3$-factor from $\frac{|N_G(u)|}{\sqrt{\Delta}}$. So, the final multiplicative error will be $3^{O(\log \log \Delta)} = \poly \log \Delta$.

\begin{lemma}\label{lem:sparsification}
     Let $G$ be a graph with bipartition $V(G) = U \sqcup V$. Let $\Delta$ and $\frac{\Delta}{f}$ be an upper bound on the maximum degree and a lower bound on the minimum degree, respectively, of any node in $U$ for any parameter $f \le \frac{\Delta^{0.4}}{\log n}$ and $f \ge \poly(\log n)$.
     There exists a sublinear MPC algorithm that computes in $O(\log \log \Delta)$ rounds a subset $V^{sub} \subseteq V$ such that for any node $u \in U$ with $\deg_G(u) \geq \frac{\Delta}{f}$, it holds that $|N_G(u) \cap V^{sub}| \in [1, 2^{O(\log f)}]$. The algorithm global space usage is linear in the input size.
\end{lemma}
\begin{proof}
    Our goal is to find a suitable set $V^{sub}$ by applying the sparsification outlined in \cref{lem:bi-partite-small-degree}. If $\Delta \ge n^\alpha$, we first apply \cref{lem:one-sampling-high-degr} for $O(1/\eps) = O(1)$ times until the maximum degree in $U$ is within the memory capacity of a single machine $O(n^\alpha)$, which can be achieved by setting $\eps \le \frac{\alpha}{10}$, i.e., $n^\alpha \ge n^{10 \eps}$. Define $\Delta' \le n^\alpha$ as the maximum degree in $U$ after downsampling vertices in $V$  for $O(1)$ iterations as per \cref{lem:one-sampling-high-degr}. Notice that the minimum degree in $U$ is now $c \cdot \frac{\Delta'}{f}$, for some constant $c > 0$. Then, we run the algorithm of \cref{lem:bi-partite-small-degree} for $k = O(\log \log \Delta)$ iterations, and stop as soon as the minimum degree in $U$ is within $2^{O(\log f)}$. We prove by induction that after k iterations nodes have degrees in 
    \begin{equation*}
        \left[\frac{c}{f \cdot 3^k} (\Delta')^{1/2^k}, \, (\Delta')^{1/2^k}\right].
    \end{equation*}
    The base case follows from \cref{lem:bi-partite-small-degree}. The induction step then follows from $$\left[\frac{c}{f \cdot 3^{(k-1)}} (\Delta')^{1/2^{(k-1)}} \cdot \frac{1}{3 (\Delta')^{1/2^{k}}},\, (\Delta')^{1/2^{(k-1)}} \cdot \frac{1}{(\Delta')^{1/2^{k}}} \right] = \left[\frac{c}{f \cdot 3^k} (\Delta')^{1/2^k}, \, (\Delta')^{1/2^k}\right].$$
    By choosing $k = \floor{\log \log \Delta' - \log (2 \log (f \cdot \log \Delta'))}$, one can verify that, for any vertex in $U$, the minimum degree in the downsampled graph will be at least one, and the maximum degree at most $2^{O(\log (f \cdot \log \Delta))} = 2^{O(\log f)}$.
\end{proof}

Our $2$-ruling set algorithm is paramterized by $f = 2^{\sqrt{\log \Delta}}$. On a high-level, we mimic the randomized local $2$-ruling set algorithm of \cite{KP12}. In each iteration $i$, $0 \le i \le \lfloor \log f \rfloor$, we address the set of vertices with degree in $(\Delta/f^{i+1}, \Delta/f^i]$. We apply the sparsification of \cref{lem:sparsification} on each set of high-degree vertices, one set at a time sequentially. Each sparsified subgraph is then put aside and, together with all incident nodes in $G$, is removed from further consideration before starting the next iteration. At the end, the union of all subgraphs of induced maximum degree $2^{O(\log f)}$ and possibly some remaining low-degree vertices are given in input to an MIS algorithm, whose solution is effectively a $2$-ruling set. We detail the algorithm in the following pseudocode and proceed to its analysis below.

\begin{algorithm}
    \caption{\textsc{Sublinear 2-Ruling Set}}
    \begin{algorithmic} 
        \State $f \leftarrow 2^{\sqrt{\log \Delta}}$; $M \leftarrow \emptyset$
		\For{$i \leftarrow 0, 1, \cdots, \floor{\log f}$}
            \State $U \leftarrow \{v \in V \mid \deg_G(v) \in (\frac{\Delta}{f^{i+1}},\, \frac{\Delta}{f^{i}}]\}$; $V' \leftarrow V$
            \State $G' \leftarrow (U \sqcup V', E' = \{(u, v) \mid u \in U, v \in V', (u, v) \in E\})$ \Comment{Bipartition for sparsification}
                \For{$j \leftarrow 1, 2, \cdots, O(\log \log \Delta)$} \Comment{See also  \cref{lem:sparsification}}
                    \State $\Delta' \leftarrow $ maximum degree in $G'$
		          \State $V' \leftarrow $ sample $v \in V'$ with prob.\ $\max\{\frac{2}{3 \sqrt{\Delta'}}, \frac{1}{n^{\eps}}\}$
                \EndFor
            \State $M \leftarrow M \cup V'$ 
            \State $V \leftarrow V \setminus (V' \cup N_G(V'))$ \Comment{Remove neighbors of sampled set}
         \EndFor
         \State Return MIS on $G[M \cup V]$
    \end{algorithmic}
\end{algorithm}

\begin{lemma}\label{lem:deg-red}
    At the end of iteration $i$, $1 \le i \le \floor{\log f}$, all vertices still in $V$ have degree at most $\max\{\frac{\Delta}{f^i}, 2^{O(\log f)}\}$.
\end{lemma}
\begin{proof}
    Consider a high-degree vertex $u \in U$ at the start of the $i$-th iteration. By \cref{lem:sparsification}, each node in $U$ is incident to a node that joins the set $M$ by the end of this iteration. Since all vertices incident to $M$ are removed from $V$, the lemma follows.
\end{proof}

\begin{lemma}\label{lem:mis_graph}
    After $\floor{\log f}$ iterations, the subgraph induced by $M$ together with vertices still in $V$, i..e, $G[M \cup V]$, has maximum degree $2^{O(\log f)}$.
\end{lemma}
\begin{proof}
    First, consider a vertex $v$ that joins the set $M$ at some iteration $j$. Observe that no neighbor of $v$ in $G$ had joined $M$ earlier, otherwise, $u$ would have been removed. By \cref{lem:sparsification}, all vertices that join $M$ at iteration $j$ have induced degree at most $2^{O(\log f)}$. Then, the neighbors of $M$ are removed from $V$ and, thus, cannot join $M$ anymore. This proves that vertices in $M$ have degree at most $2^{O(\log f)}$.
    Second, consider a vertex $w$ that at the end of the $\floor{\log f}$-th iteration is still in $V$. This means that $w$ does not neighbor $M$ and that, by \cref{lem:deg-red}, $w$ has degree at most $2^{O(\log f)}$, finishing the claim.
\end{proof}

\begin{proof}[Proof of \cref{thm:sublinear-memory}]
    As proved in \cref{lem:sparsification}, each iteration of the algorithm runs in $O(\log \log \Delta)$ rounds. Since there are $O(\sqrt{\log \Delta})$ iterations for $f = 2^{\sqrt{\log \Delta}}$, the total number of rounds is $O(\sqrt{\log \Delta} \cdot \log \log \Delta)$. From \cref{lem:mis_graph}, we see that the sparsified graph given by $M$ together with vertices still in $V$ has degree at most $2^{O(\sqrt{\log \Delta})}$. Therefore, the MIS computation at the end of the algorithm takes $O(\sqrt{\log \Delta} + \log \log^* n)$ by using the deterministic MIS algorithm from Lemma 27 of \cite{CDP24} that runs in $O(\log \Delta' + \log \log^* n)$ on a $\Delta'$-maximum degree graph, provided that the allowed global space is $O(n^{1+\delta}+m)$. Otherwise, we use the variation given in \cite{FGG23} that runs in $O(\sqrt{\log \Delta} \cdot \log \log n)$ and uses linear global space.
\end{proof}

Lastly, we need to show how to achieve a $\poly(\Delta)$ coloring of $G^2$ to fulfill the assumption made in \cref{lem:bi-partite-small-degree}.

Whenever $\Delta = n^{\Omega(1)}$, the initial assignment of IDs to vertices, typically from $1$ to $n$, effectively serves as a $\poly(\Delta)$ coloring of $G^2$. In the case where $\Delta \le n^{\delta}$ for constant $\delta < \alpha/2$, we ensure $\Delta^2 \ll n^{\alpha}$. This implies that the $2$-hop neighborhood of every node can be stored within the local memory of a single machine. Storing the $2$-hop neighbors on a single machine permits the use of Linial's coloring reduction technique \cite{Lin92}, which achieves a $O(\Delta^6)$ coloring in $O(1)$ rounds. However, this approach necessitates of a global space usage of $O(n^{1+2\delta})$, potentially exceeding $O(n+m)$.
To improve the global space usage, after three runs of \cref{lem:bi-partite-small-degree}, the degree of each vertex which has not been removed is at most $\Delta^{0.22}$. Since each sampled vertex is incident to a high-degree vertex of initial degree at least $O(\Delta/f)$, we can charge high-degree vertices $O(\Delta^{0.66}) \ll \Delta/f$ space consumption. This reduction allows us to gather the $2$-hop neighbors of all active nodes onto single machines without breaching the global space limit.
A further optimization involves substituting the first three runs of \cref{lem:bi-partite-small-degree} with a weaker version, detailed below, addressing all but at most $\frac{n}{\Delta^{0.01}}$ vertices. The proof follows from that of \cref{lem:bi-partite-small-degree}.

\begin{lemma}\label{lem:one-step-derand-weak}
    Let $G = (V, E)$ be a graph with an upper bound $\Delta$ on the maximum degree. There is a sublinear MPC algorithm that computes in $O(1)$ rounds a subset $V' \subseteq V$ ensuring that, for all but at most $\frac{n}{\Delta^{0.01}}$ vertices $v \in V$ with $\deg_G(v) \geq \log(n) \cdot \Delta^{0.6}$, it holds that $|N_G(v) \cap V'| \in \left[\frac{1}{3\sqrt{\Delta}}|N_G(v)|, \frac{1}{\sqrt{\Delta}}|N_G(v)|\right]$.
\end{lemma}

Applying \cref{lem:one-step-derand-weak} initially and excluding up to $\frac{n}{\Delta^{\Omega(1)}}$ vertices not meeting our criteria allows for the execution of $O(\log \log \Delta)$ iterations for the well-behaved vertices. The excluded vertices are subsequently addressed by repeating the same process. After $O(1)$ iterations, the remaining vertex count drops to $O(\frac{n}{\Delta^2})$, fitting the global space needed to store their $2$-hop neighborhoods within $O(n)$. Consequently, after $O(\log \log \Delta)$ rounds, all vertices are processed without affecting the asymptotic total number of rounds.

\section*{Acknowledgements}
We thank Christoph Grunau and Manuela Fischer for helpful comments and discussions.

Jeff Giliberti gratefully acknowledges financial support by the Fulbright U.S. Graduate Student Program, sponsored by the U.S. Department of State and the Italian-American Fulbright Commission. The content does not necessarily represent the views of the Program.
\bibliographystyle{alpha}
\bibliography{ref}

\end{document}